%% file: MILPSensitivity.tex
\documentclass[journal,twoside,web]{ieeecolor}
\usepackage{lcsys}
\usepackage{cite}
\usepackage{amsmath,amssymb,amsfonts}
\usepackage{algorithmic}
\usepackage{graphicx}
\usepackage{textcomp}
\usepackage[dvipsnames]{xcolor}

\makeatletter
\let\proof\@undefined                        
\let\endproof\@undefined                  
\makeatother
\usepackage{amsthm}
\usepackage{multirow}
\usepackage{dsfont}
\DeclareFontFamily{OT1}{pzc}{}
\DeclareFontShape{OT1}{pzc}{m}{it}{<-> s * [1.10] pzcmi7t}{}
\DeclareMathAlphabet{\pzc}{OT1}{pzc}{m}{it}
\newcommand\Item[1][]{%
  \ifx\relax#1\relax  \item \else \item[#1] \fi
  \abovedisplayskip=0pt\abovedisplayshortskip=0pt~\vspace*{-\baselineskip}}

\usepackage{amssymb,mathtools} 
\usepackage{graphicx,subcaption}    
\usepackage{pgfplots}
\usepgfplotslibrary{colormaps,colorbrewer,fillbetween,patchplots}
\usetikzlibrary{shapes.multipart}
\usetikzlibrary[pgfplots.groupplots]
\newcommand{\ignore}[1]{}
\pgfplotsset{compat=newest,tick label style={font=\small},
	every axis/.append style={label style={font=\small}}}
	
\pgfdeclarelayer{background}
\pgfsetlayers{background,main}
	
\tikzset{dot/.style={circle,inner sep=1.5pt},
	nodabove/.style={font=\scriptsize,above},nodbelow/.style={font=\scriptsize,below},}

\def\BibTeX{{\rm B\kern-.05em{\sc i\kern-.025em b}\kern-.08em
		T\kern-.1667em\lower.7ex\hbox{E}\kern-.125emX}}

\newcommand{\Ccal}{\mathcal{C}}
\newcommand{\Xcal}{\mathcal{X}}
\newcommand{\Acal}{\mathcal{A}}
\newcommand{\Bcal}{\mathcal{B}}
\newcommand{\Dcal}{\mathcal{D}}

\newcommand{\Rcal}{\mathcal{R}}

\newcommand{\Ical}{\mathcal{I}}
\newcommand{\diam}{\mathrm{diam}}

\newtheorem{thm}{Theorem}

\newtheorem{lem}{Lemma}

\newtheorem{prop}{Proposition}

\newtheorem{defn}{Definition}
\newtheorem{exmp}{Example}

\title{\LARGE \bf 
Sensitivity Analysis for Piecewise-Affine Approximations of Nonlinear Programs with Polytopic Constraints}

\author{Leila Gharavi, Changrui Liu, Bart De Schutter \textit{Fellow IEEE}, Simone Baldi \textit{Senior Member IEEE}
\thanks{This research is supported by the Dutch Science Foundation NWO-TTW within the EVOLVE project (no.\ 18484), and by the 
European Research Council (ERC) under the European Union’s Horizon 
2020 research and innovation programme (Grant agreement No. 101018826 - CLariNet)(Corresponding authors: L.\ Gharavi and S.\ Baldi)}
\thanks{L.\ Gharavi, C.\ Liu and B.\ De~Schutter are with the Delft Center for Systems and Control, Delft University of Technology, Delft, The Netherlands.
        {\tt\small \{L.Gharavi,C.Liu-14,B.DeSchutter\}@tudelft.nl}}
\thanks{S.\ Baldi is with the School of Mathematics, Southeast University, Nanjing, China. {\tt\small 103009004@seu.edu.cn}}}

\begin{document}

\maketitle
\thispagestyle{empty}
\pagestyle{empty}

\begin{abstract}

Nonlinear Programs (NLPs) are prevalent in optimization-based control of nonlinear systems. Solving general NLPs is computationally expensive, necessitating the development of fast hardware or tractable suboptimal approximations.
This paper investigates the sensitivity of the solutions of NLPs with polytopic constraints when the nonlinear continuous objective function is approximated by a PieceWise-Affine (PWA) counterpart.
By leveraging perturbation analysis using a convex modulus, we derive guaranteed bounds on the distance between the optimal solution of the original polytopically-constrained NLP and that of its approximated formulation. Our approach aids in determining criteria for achieving desired solution bounds. Two case studies on the Eggholder function and nonlinear model predictive control of an inverted pendulum demonstrate the theoretical results.

\end{abstract}

\begin{IEEEkeywords}
Perturbation Analysis, Non-Convex Nonlinear Programming, Piecewise-Affine Functions, Max-Min-Plus-Scaling Systems, Function Approximation
\end{IEEEkeywords}

\section{Introduction}

\IEEEPARstart{N}{onlinear} Programs (NLPs) are commonly encountered in optimization-based control of nonlinear systems, e.g., Nonlinear Model Predictive Control (NMPC)~\cite{Gros2020}. 
Solving non-convex NLPs is intractable, posing a great challenge in applying optimization-based controllers in real-time operations, especially for systems having fast dynamics.
Various solutions have been proposed to address this issue, such as adaptive problem formulations~\cite{Griffith2018}, learning-based methods~\cite{Dong2019}, and 
sensitivity analysis of NLPs~\cite{Zavala2009}.

PWA approximations are widely used due to their tractability~\cite{Jin2021,Bemporad2011}. To obtain a continuous PWA approximation, $\min$ and $\max$ operators can be used to 
maintain continuity and to resolve numerical issues in the resulting optimization problem~\cite{Xu2012,Deschutter2004,PSMPC}. The approximated problem can be used to obtain a suboptimal solution~\cite{Chen2019}, whose optimality highly depends on the accuracy of the approximation. For example, a warm start of a non-convex NLP can be obtained by solving the approximated optimization problem~\cite{Lawrynczuk2022,Ghandriz2023}. Optimality guarantees of such approaches can be derived using sensitivity analysis, establishing an upper bound on the distance between the original solutions and the approximated ones. As a result, by finding a subset of the decision space around the approximated solution, one can sample a structured or random warm start to solve the original non-convex NLP more efficiently.  

Quantitative bounds on the distance between the original and the approximated solutions have been studied in the sensitivity analysis of quadratic~\cite{Daniel1973} and convex~\cite{Phu2010} optimization problems. Regarding NLPs, there exist several results on their sensitivity to the parameters in the optimization formulation~\cite{Grotschel2001,Subotic2022} and the initial solution~\cite{Zavala2023}. In addition, optimality and dissipativity conditions for the perturbed convex NLP problem have also been established~\cite{Pirnay2012,Muller2015}. For a more extensive study, the reader can refer to~\cite{Puranik2017}. Recently, sensitivity analysis has also been applied to analyze the infinite-horizon performance of MPC~\cite{liu2024stability}. However, obtaining quantitative bounds on the distance between the solutions of a non-convex NLP and its PWA approximation is still a gap that needs to be filled, and our work addresses this problem.

In this paper, we present a method to bound the solutions of a polytopically-constrained non-convex NLP problem using a continuous PWA approximation of the nonlinear objective function. We employ the Max-Min-Plus-Scaling (MMPS) formulation of continuous PWA functions in~\cite{Deschutter2004} to construct a piecewise convex approximation formalism. Leveraging some results in~\cite{Phu2010}, we derive guaranteed bounds on the distance between the original and the approximated optimal solution. Moreover, our approach can not only establish such bounds but also aid in determining necessary criteria during the approximation stage to attain a desired solution bound. To demonstrate the theoretical findings, we present two case studies on the Eggholder function~\cite{Whitley1996}, a renowned benchmark optimization problem with multiple local minima, and an NMPC optimization problem for an inverted pendulum~\cite{Gros2020}.

The rest of the paper is organized as follows. In Section~\ref{sec:back}, we present preliminaries regarding the sensitivity analysis of NLPs and the PWA approximation of nonlinear functions. Section~\ref{sec:approx} formulates the problem and Section~\ref{sec:prob} elaborates our proposed approach to theoretically compute the confidence radius for the local minima of the corresponding approximated function. In Section~\ref{sec:case}, we then demonstrate the derived confidence radius through a case study on the Eggholder function and we apply our analysis to an NMPC optimization example. Section~\ref{sec:conc} concludes this paper.

\emph{Notation:} For a positive integer $P$, we use $\Ical_P$ to denote the set $\{1, 2, \dots, P\}$. For a connected set $\Dcal \subseteq \mathbb{R}^n$, the diameter of $\Dcal$ is defined as $\diam(\Dcal) \coloneqq \max_{x_1,x_2 \in \Dcal} \Vert x_1 - x_2 \Vert$, where $\|\cdot\|$ is the Euclidean or the 2-norm.

\newpage 
\section{Preliminaries}\label{sec:back}

\subsection{Representation of Continuous PWA Functions}

We start with formally defining a continuous PWA function which will be frequently used throughout the paper.
\begin{defn}[Continuous PWA function \cite{Chua1988}] A scalar-valued function \mbox{$f : \Dcal \subseteq \mathbb{R}^n \to \mathbb{R}$} is said to be a continuous PWA function if and only if the following conditions hold:
	\begin{enumerate}
		\item the domain space $\Dcal$ is divided into a finite number of closed polyhedral regions $\Rcal_{1}, \dots, \Rcal_{R}$ with non-overlapping interiors,
		\item for each $r \in \Ical_R$, $f$ can be expressed as 
		\[f (x) = \alpha_{r}^T x + \beta_r \quad \rm{if} \quad x \in \Rcal_r,\]
		with $\alpha_{r} \in \mathbb{R}^n$ and $\beta_r \in \mathbb{R}$, and
		\item $f$ is continuous on the boundary between any pair of regions.
	\end{enumerate}	\label{def:cpwa}
\end{defn}
PWA functions can be expressed in different forms, among which MMPS form is powerful for decomposing PWA functions.
\begin{thm}[MMPS representation \cite{Deschutter2004}] For a scalar-valued continuous PWA function $f$ as in Definition~\ref{def:cpwa}, there exist non-empty index sets $\Ical_P$ and $\Ical_{Q_p}$ such that
	\begin{equation}
		f (x) = \min_{p \in \Ical_P} \max_{q \in \Ical_{Q_p}} \left(a_{p,q}^T x + b_{p,q}\right),\label{eq:fmmps}
	\end{equation}
	for 
	real numbers $b_{p,q}$ and vectors $a_{p,q} \in \mathbb{R}^n$.
\end{thm}
For convenience, we define the $p$-th local convex segment of $f$ as
\begin{align}
	f_p (x) \coloneqq \max\limits_{q \in \Ical_{Q_p}} \left(a_{p,q}^T x + b_{p,q}\right),\label{eq:fp}
\end{align}
where $f_p$ is convex since it is defined as the maximum of a finite number of affine functions and its domain is also convex. In addition, we define the region $\Ccal_{p,q}$ in which a certain affine function is activated and the region $\Ccal_{p,.}$ in which a convex PWA function is activated, that is,
\begin{subequations}
	\begin{align}
		\Ccal_{p,q} & \coloneqq \{ x \in \Dcal \; | \; f (x) = a_{p,q}^T x + b_{p,q}\},\\
		\Ccal_{p,.} & \coloneqq \{ x \in \Dcal \; | \; f (x) = f_p(x)\}.
	\end{align}
\end{subequations}
Further, we have
{\mbox{$\Ccal_{p,.} = \bigcup_{q=1}^{Q_p} \Ccal_{p,q}.$}
	Lastly, we define \mbox{$\sigma_p : \Ccal_{p,.} \to \Ical_{Q_p}$} as the region index function for $f_p$ as
	\begin{align}
		\sigma_p (x) = q \iff x \in \Ccal_{p,q}. \label{eq:sigma}
	\end{align}

	\subsection{Sensitivity of the Convex Optimization Problem}
	
	The convexity modulus~\cite{Phu2010}, being used to quantify convexity, is useful in the sensitivity analysis of convex functions. In the following, we define the convexity modulus specifically for $f_p$, the $p$-th convex segment of $f$, and its domain $\Ccal_{p,.}$.
	\begin{defn}[Convexity modulus \cite{Phu2010}] For a given convex function $f_p$, the convexity modulus \mbox{$h_1 : [0,+\infty) \to [0,+\infty)$} over the domain $\Ccal_{p,.}$ is defined as
		\begin{align}
			h_1 (\gamma) \coloneqq  
			\begin{cases}
				\begin{array}{ll}
					\inf\limits_{\substack{v,w \in \Ccal_{p,.}\\\Vert v - w \Vert = \gamma}} J (v,w) & \text{if }\gamma < \diam (\Ccal_{p,.})\\
					+\infty & \text{if }\gamma \geqslant \diam (\Ccal_{p,.})
				\end{array}			
			\end{cases},
			\label{eq:h1def}
		\end{align}
		where $v$ and $w$ are two points in $\Ccal_{p,.}$ satisfying \mbox{$\Vert v - w \Vert = \gamma$}, and $J(v, w)$ is given as
		\begin{align}
			J (v,w) = \dfrac{f_p(v) + f_p(w)}{2}  - f_p \left(\dfrac{v+w}{2}\right).\label{eq:jvw}
		\end{align}\label{def:h1}
	\end{defn}
	
	\begin{thm}[Theorem 4.5 in \cite{Phu2010}]\label{thm:phu} Suppose that \mbox{$f_p : \Ccal_{p,.} \to \mathbb{R}$} is a scalar-valued convex function and \mbox{$\delta_p : \Ccal_{p,.} \to \mathbb{R}$} is an arbitrary function satisfying
		\begin{align}
			\sup_{x \in \Ccal_{p,.}} \vert \delta_p(x) \vert = \Delta_p < \infty.\label{eq:deltap}
		\end{align}
		Let $x_p^\ast$ be any global infimizer of $f_p$ and $\hat{x}_p^\ast$ be any global infimizer of $\hat{f}_p = f_p + \delta_p$. Then
		\begin{align}
			\Vert \hat{x}_p^\ast - x_p^\ast \Vert \leqslant h_1^{-1} \left(2 \Delta_p \right),
			\label{eq:bound}
		\end{align}
		where $h_1$ is the convexity modulus in Definition~\ref{def:h1}.
	\end{thm}
	
	For a more compact notation, we call the right-hand side of (\ref{eq:bound}), the confidence radius, defined as follows:
	\begin{defn}[Confidence radius] For a given function \mbox{$\hat{f}_p : \Ccal_{p,.} \subseteq \mathbb{R}^n \to \mathbb{R}$} approximating the function \mbox{$f_p : \Ccal_{p,.} \subseteq \mathbb{R}^n \to \mathbb{R}$} with the maximal approximation error $\Delta_p$ from (\ref{eq:deltap}), the confidence radius is the upper-bound on the distance between $\hat{x}^\ast_p$, the global minimizer of $\hat{f}_p$, and $x^\ast_p$, the global minimizer of $f_p$, and is obtained by
		\[\chi = h_1^{-1} (2 \Delta_p),\]
		where $h_1$ is the convexity modulus in Definition~\ref{def:h1}.\label{def:chi}
	\end{defn}
	
	\begin{prop}[Proposition 2.2, 2.5 in \cite{Phu2010}] Given a convex function $f_p$ on a compact domain $\Ccal_{p,.}$ with convexity modulus $h_1$ defined as (\ref{eq:h1def}), the following hold:
		\begin{enumerate}
			\item $h_1$ is left-continuous on \mbox{$\left(0,\diam(\Ccal_{p,.})\right)$}, and
			\item for $\gamma_1,\gamma_2 \in [0,\diam(\Ccal_{p,.})$, if \mbox{${\gamma_1 < \gamma_2}$}, then $h_1 (\gamma_1) / \gamma_1 \leqslant h_1 (\gamma_1) / \gamma_2$.
		\end{enumerate} \label{prop:h1prop}
	\end{prop}

\section{Continuous PWA Approximation of NLPs}\label{sec:approx}

Consider an NLP with polytopic constraints
\begin{align} \label{eq:nlp}
	\min_{x \in \Xcal} & \quad F (x),
\end{align}
where \mbox{$F : \Dcal \subset \mathbb{R}^n \to \mathbb{R}$} is the nonlinear objective function and $\Xcal \subseteq \Dcal$ is the polytopic feasible region. From now on, we assume that the domain $\Dcal$ is compact. We approximate $F$ by a continuous PWA function $f$ of the MMPS form (\ref{eq:fmmps}) via solving the approximation problem
\begin{align} \label{eq:approxprob}
	\min_{\Acal, \ \Bcal} & \quad \int\limits_{\Dcal} \left\vert F(x) - \min_{p \in \Ical_P} \max_{q \in \Ical_{Q_p}} \left(a_{p,q}^T x + b_{p,q}\right) \right\vert \; d x,
\end{align}
to minimize the absolute approximation error, where the ordered sets $\Acal$ and $\Bcal$ respectively collect $a_{p,q}$ and $b_{p,q}$.
\begin{exmp}
	Figure~\ref{fig:fapprox} shows a 1-dimensional example of approximating a nonlinear objective function $F$ by a continuous PWA function $f$ using the MMPS form (\ref{eq:fmmps}) as
	\[F (x) \approx f (x) = \min \left( \underbrace{\max \left( f_{1,1} , f_{1,2} \right)}_{f_{1,.}},  \underbrace{\max \left( f_{2,1} , f_{2,2} \right)}_{f_{2,.}} , f_{3,.} \right),\]
	with $P = 3$, $Q_1 =Q_2 = 2$, and \mbox{$Q_3$ = 1}. The convex segments of $f$ are shown by $f_{p,.}$ which give the maximum value among $Q_p$ affine functions $f_{p,q}$, $q \in \Ical_{Q_p}$. The subregions $\Ccal_{p,q}$ are shown in the same color as their corresponding active affine functions, $f_{p,q}$. In this 1-dimensional example, $\diam(\Ccal_{p,.})$ is the distance between the upper and lower bounds of $\Ccal_{p,.}$ on the $x$-axis.\vspace{-10pt}
	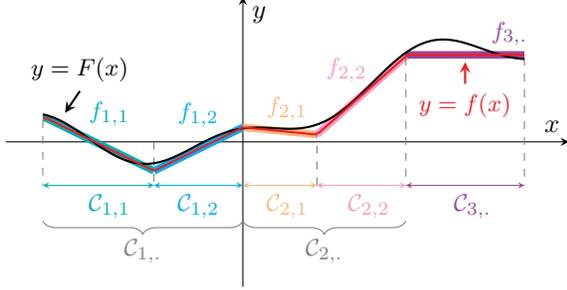
\begin{figure}[htbp]\centering
		\begin{tikzpicture}
			\begin{axis}[height=0.3\textwidth,width=0.5\textwidth,ytick=0,xtick=\empty,
				xlabel=$x$,ylabel=$y$,xmin=-1.6,xmax=2.2,ymin=-1.0,ymax=1.0,
				legend pos=south east,axis x line=middle,axis y line=middle,]		
				\addplot[color=TealBlue,opacity=0.7,line width=1mm,samples=100,domain=-1.35:-0.6]
				{max(1.0*x-0.1,0.4,-1*x-0.2)-max(1.0*x-0.7,0.1*x+0.3,-0.5*x+0.3,-1.1*x-0.5)};
				\addplot[color=Cyan,opacity=0.7,line width=1mm,samples=100,domain=-0.6:0]
				{max(1.0*x-0.1,0.4,-1*x-0.2)-max(1.0*x-0.7,0.1*x+0.3,-0.5*x+0.3,-1.1*x-0.5)};
				\addplot[color=Apricot,opacity=0.7,line width=1mm,samples=100,domain=0:0.5]
				{max(1.0*x-0.1,0.4,-1*x-0.2)-max(1.0*x-0.7,0.1*x+0.3,-0.5*x+0.3,-1.1*x-0.5)};
				\addplot[color=Lavender,opacity=0.7,line width=1mm,samples=100,domain=0.5:1.1]
				{max(1.0*x-0.1,0.4,-1*x-0.2)-max(1.0*x-0.7,0.1*x+0.3,-0.5*x+0.3,-1.1*x-0.5)};
				\addplot[color=Purple,opacity=0.7,line width=1mm,samples=100,domain=1.1:1.9]
				{max(1.0*x-0.1,0.4,-1*x-0.2)-max(1.0*x-0.7,0.1*x+0.3,-0.5*x+0.3,-1.1*x-0.5)};
				\addplot[color=gray,dashed] coordinates {(-1.35,0.15)(-1.35,-0.3)}; 
				\addplot[color=gray,dashed] coordinates {(-0.6,0)(-0.6,-0.3)};  
				\addplot[TealBlue,ultra thin, stealth-stealth] coordinates {(-1.35,-0.3)(-0.6,-0.3)};
				\node[color=TealBlue,below] at (-0.9,0.-0.3){\small{$\Ccal_{1,1}$}};
				\node[color=TealBlue,above] at (-0.9,0.05){\small{$f_{1,1}$}};
				\addplot[color=gray,dashed] coordinates {(0,0.1)(0,0)};
				\addplot[Cyan,ultra thin, stealth-stealth] coordinates {(-0.6,-0.3)(0,-0.3)};
				\node[color=Cyan,below] at (-0.3,-0.3){\small{$\Ccal_{1,2}$}};
				\node[color=Cyan,above] at (-0.3,0.04){\small{$f_{1,2}$}};
				\draw [gray,decorate,decoration={brace,amplitude=5pt,mirror,raise=3ex}](-1.35,-0.3) -- (0,-0.3) node[midway,yshift=-2.5em]{\small{$\Ccal_{1,.}$}};
				\addplot[color=gray,dashed] coordinates {(0.5,-0.3)(0.5,0.05)};
				\addplot[Apricot,ultra thin, stealth-stealth] coordinates {(0,-0.3)(0.5,-0.3)};
				\node[color=Apricot,below] at (0.3,-0.3){\small{$\Ccal_{2,1}$}};
				\node[color=Apricot,above] at (0.3,0.08){\small{$f_{2,1}$}};
				\addplot[color=gray,dashed] coordinates {(1.1,0.58)(1.1,-0.3)};
				\addplot[Lavender,ultra thin, stealth-stealth] coordinates {(0.5,-0.3)(1.1,-0.3)};
				\node[color=Lavender,below] at (0.85,-0.3){\small{$\Ccal_{2,2}$}};
				\node[color=Lavender,above] at (0.7,0.35){\small{$f_{2,2}$}};
				\draw [gray,decorate,decoration={brace,amplitude=5pt,mirror,raise=3ex}](0,-0.3) -- (1.1,-0.3) node[midway,yshift=-2.5em]{\small{$\Ccal_{2,.}$}};
				\addplot[color=gray,dashed] coordinates {(1.9,0.58)(1.9,-0.3)};
				\addplot[Purple,ultra thin, stealth-stealth] coordinates {(1.1,-0.3)(1.9,-0.3)};
				\node[color=Purple,below] at (1.55,-0.3){\small{$\Ccal_{3,.}$}};
				\node[color=Purple,above] at (1.8,0.6){\small{$f_{3,.}$}};
				\addplot[color=black, thick,samples=100,domain=-1.35:1.9]
				{0.4*sin(deg(0.6*x+0.8)) + 0.3*sin(deg(1.6*x-1.3)) + 0.1*sin(deg(4.9*x+1.9))};
				\addplot[black, thick, stealth-] coordinates {(-1.2,0.20) (-1.1,0.35) }
				node[above,font=\small] {$y = F (x)$};
				\addplot[color=Red, thick,samples=100,domain=-1.35:1.9]
				{max(1.0*x-0.1,0.4,-1*x-0.2)-max(1.0*x-0.7,0.1*x+0.3,-0.5*x+0.3,-1.1*x-0.5)};
				\addplot[color=Red,thick, stealth-] coordinates {(1.5,0.55) (1.5,0.4) } 
				node[below,font=\small]
				{$y = f (x)$};
			\end{axis}
		\end{tikzpicture}\caption{A conceptual example of approximating a nonlinear function $F$ with a continuous PWA approximation $f$ using the MMPS form in (\ref{eq:fmmps}).}\label{fig:fapprox}
	\end{figure}
\end{exmp}

\section{Theoretical Analysis}\label{sec:prob}
The main result of this section is the sensitivity bound in Theorem \ref{thm:main}. Before, we present some building blocks about continuity and boundedness of the convexity modulus (cf. Lemma \ref{lem:dh1pwc}, \ref{lem:dh1cont} and \ref{lem:h1form}).

\begin{lem}
	For a convex PWA function $f_p$ expressed by (\ref{eq:fp}), $\partial h_1 / \partial \gamma$ is a piecewise constant function.\label{lem:dh1pwc}
\end{lem}
\begin{proof}
	For compactness, let us denote
	\[	i = \sigma_p (v), \qquad 	j = \sigma_p (w), \qquad k = \sigma_p \left(\dfrac{v+w}{2}\right),\]
	with $i,j,k \in \Ical_{Q_p}$. For $f_p$ expressed by (\ref{eq:fp}), the function $J$ in (\ref{eq:jvw}) can be written as
	\begin{multline}
		J(v,w) = \\ \dfrac{a_{p,i}^T - a_{p,k}^T}{2} v + \dfrac{a_{p,j}^T - a_{p,k}^T}{2} w  +  \dfrac{b_{p,i}+b_{p,j} - 2 b_{p,k}}{2}.\label{eq:jvwnew}
	\end{multline}
	The necessary Lagrange conditions for optimality at $(v^\ast, w^\ast)$ in (\ref{eq:h1def}) state that there must exist $\mu \in \mathbb{R}$ that satisfies the following simultaneously:
	
	{\small
		\begin{subequations}\label{eq:jgrads}
			\begin{align}
				& \nabla J (v ^\ast, w^\ast) + \mu \nabla \left(\Vert v - w \Vert - \gamma \right)\Big|_{\substack{v = v^\ast\\w = w^\ast}}  = 0, \\
				&\Vert v^\ast - w^\ast \Vert = \gamma.
			\end{align}
	\end{subequations}}%
	By calculating the gradient of (\ref{eq:jvwnew}), we have that (\ref{eq:jgrads}) becomes
	{\small
		\begin{subequations}
			\begin{align}
				&\dfrac{a_{p,i}^T - a_{p,k}^T}{2} + \mu \dfrac{v^\ast - w^\ast}{\Vert v^\ast - w^\ast \Vert} = 0, \\
				&\dfrac{a_{p,j}^T - a_{p,k}^T}{2} + \mu \dfrac{w^\ast - v^\ast}{\Vert v^\ast - w^\ast \Vert} = 0,\\
				&\Vert v^\ast - w^\ast \Vert = \gamma,
			\end{align}
	\end{subequations}}%
	which implies the existence of $\mu \in \mathbb{R}$ satisfying
	\begin{align}
		\dfrac{a_{p,i}^T - a_{p,j}^T}{2} + 2 \mu \dfrac{v^\ast - w^\ast}{\gamma} = 0.\label{eq:exmu}
	\end{align}
	Note that $\partial h_1 \left/ \partial \gamma \right.$ is equal to $\partial J \left/ \partial \gamma \right.$, except where the indices $i$, $j$, and $k$ change. At such points, $h_1$ is not differentiable with respect to $\gamma$, which does not conflict with $\partial h_1 \left/ \partial \gamma \right.$ being a piecewise constant function. To find the slope of $h_1$ where it exists, the chain rule can be applied as
	\begin{align*}
		\dfrac{\partial h_1}{\partial \gamma} = \dfrac{\partial J}{\partial v^\ast} \left/ \dfrac{\partial \gamma}{\partial v^\ast}\right. +\dfrac{\partial J}{\partial w^\ast} \left/ \dfrac{\partial \gamma}{\partial w^\ast}\right. = \dfrac{a_{p,i}^T - a_{p,j}^T}{2} \dfrac{\Vert v^\ast - w^\ast \Vert}{v^\ast - w^\ast},
	\end{align*}
	which, considering (\ref{eq:exmu}), leads to
	\begin{align}
		\dfrac{\partial h_1}{\partial \gamma} = - 2 \mu,\label{eq:dh1}
	\end{align}
	which implies that $\partial h_1 / \partial \gamma$ is a function of the $a_{p,i}^T - a_{p,j}^T$ values.
\end{proof}
\begin{lem}\label{lem:dh1cont}
	For a convex PWA function $f_p$ expressed by (\ref{eq:fp}), the convexity modulus $h_1$ in (\ref{eq:h1def}) is continuous on $[0,\diam(\Ccal_{p,.}))$.
\end{lem}
\begin{proof}
	From Proposition~\ref{prop:h1prop}, we know that $h_1$ is left-continuous on $[0,+\infty)$. Seeking a contradiction, let us assume that $h_1$ is not right-continuous in $\gamma_0 \in [0,\diam(\Ccal_{p,.}))$, hence,
	\[\lim_{\gamma \to \gamma_0^+} h_1 (\gamma) \neq h_1 (\gamma_0).\]
	The monotonicity property of $h_1$ in Proposition~\ref{prop:h1prop} requires 
	\begin{align*}
		\lim_{\gamma \to \gamma_0^+} h_1 (\gamma) > h_1 (\gamma_0).
	\end{align*}
	Therefore, without loss of generality, we assume there exists a gap $\epsilon_0 > 0$ and a point $\gamma_0 < \gamma_0^+ < \diam(\Ccal_{p,.})$ such that
	\begin{align}
		h_1 (\gamma_0^+) = h_1 (\gamma_0) + \dfrac{\partial h_1}{\partial \gamma} \Big|_{\gamma = \gamma_0} \left(\gamma_0^+ - \gamma_0\right) + \epsilon_0.\label{eq:h1conteps0}
	\end{align}
	Using (\ref{eq:h1def}), we define the points $v_0$, $w_0$, and $w_0^+$ such that
	\begin{align*}
		h_1 (\gamma_0) = \inf\limits_{\substack{v,w \in \Ccal_{p,.}\\\Vert v - w \Vert = \gamma_0}} J (v,w)  = J (v_0, w_0),
	\end{align*}
	and $\Vert v_0 - w_0^+ \Vert = \gamma_0^+$.	Considering the optimality property in (\ref{eq:h1def}), we have
	\begin{align*}
		h_1 (\gamma_0^+) \leqslant J (v_0, w_0^+),
	\end{align*}
	and knowing that $J$ is a continuous function by definition, we can deduce
	\begin{align*}
		h_1 (\gamma_0^+) \leqslant h_1 (\gamma_0) + \rho (w_0^+ - w_0),
	\end{align*}
	where $\rho : \Ccal_{p,.} \to \mathbb{R}$ is a function with the following property:
	\begin{align}
		\lim_{\nu \to 0^+} \rho (\nu) = 0.\label{eq:h1b4lim}
	\end{align}
	Substituting (\ref{eq:h1conteps0}) into (\ref{eq:h1b4lim}) and taking the limit on both sides when $\gamma_0^+$ approaches $\gamma_0$ leads to 
	\begin{align*}
		h_1 (\gamma_0) + \epsilon_0 \leqslant h_1 (\gamma_0),
	\end{align*}
	which contradicts the fact that $\epsilon_0 > 0$. Therefore, $h_1$ is right-continuous on $[0, \diam (\Ccal_{p,.}))$.
\end{proof}

\begin{prop}\label{prop:h1is0}
	For a convex PWA function $f_p$ expressed by~(\ref{eq:fp}), we have 
	\[h_1 (\gamma) = 0, \qquad \forall \gamma \leqslant \max\limits_{q \in \Ical_{Q_p}} \left\{ \diam(\Ccal_{p,q})\right\}.\]
\end{prop}
\begin{proof}
	By seeking a contradiction, let us assume that
	\[\exists \gamma_0 \leqslant \max\limits_{q \in \Ical_{Q_p}} \left\{ \diam(\Ccal_{p,q})\right\}, \quad \text{such that} \quad h_1 (\gamma_0) > 0,\]
	with the corresponding optimal points $v^\ast_0$ and $w^\ast_0$ from (\ref{eq:h1def}) such that
	\begin{align*}
		h_1 (\gamma_0) = \inf\limits_{\substack{v,w \in \Ccal_{p,.}\\\Vert v - w \Vert = \gamma_0}} J (v,w)  = J (v^\ast_0, w^\ast_0).
	\end{align*}
	Let us select two points, $v_0$ and $w_0$, on the largest subregion in $\Ccal_{p,.}$ such that
	\[\Vert v_0 - w_0 \Vert = \gamma_0.\]
	which results in $J (v_0, w_0) = 0$. Considering the optimality property in (\ref{eq:h1def}), we have
	\[J (v^\ast_0,w^\ast_0) \leqslant J (v_0,w_0),\]
	which contradicts the initial assumption that $h_1 (\gamma_0) > 0$ and $ h_1 (\gamma_0)= J (v^\ast_0,w^\ast_0)$.
\end{proof}

\begin{lem}\label{lem:h1form}
	For a convex PWA function $f_p$ expressed by (\ref{eq:fp}), the convexity modulus $h_1$ 
	is bounded by $\hat{h}_1 \leqslant h_1$, with
	\begin{align}
		\hat{h}_1 (\gamma) \coloneqq  
		\begin{cases}\begin{array}{ll}
				0 & \text{if }\gamma < \diam (\Ccal_{p,.})\\
				c_1 \gamma + c_0 & \text{if }\gamma \geqslant \diam (\Ccal_{p,.})
			\end{array}			
		\end{cases},\label{eq:h1hat}
	\end{align}
	where
	\begin{subequations}
		\begin{align}
			c_1 = \; \min\limits_{j \in \Ical_{Q_p}} \quad & \left\{ \dfrac{a_{p,i}^T - a_{p,j}^T}{2} \right\}, \\
			\mathrm{s.t.} \quad & i = \arg\max\limits_{q \in \Ical_{Q_p}} \diam(\Ccal_{p,q}), \\
			& \Ccal_{p,i} \cap \Ccal_{p,j} \neq \emptyset, 
		\end{align}\label{eq:c1def}
	\end{subequations}
	and $c_0 = \; c_1 \; \diam (\Ccal_{p,.})$.		
\end{lem}
\begin{proof}
	This can be directly deduced from Proposition~\ref{prop:h1is0}, considering the continuity of $h_1$ from Lemma~\ref{lem:dh1cont}, the piecewise-constant property of $\partial h_1 / \partial \gamma$ from Lemma~\ref{lem:dh1pwc}, and the increasing property of $h_1$ from Proposition~\ref{prop:h1prop}.
\end{proof}
\input{pics/pic_Fandf}
We are now in the position to state our main result:
\begin{thm}\label{thm:main}
	Let $F : \Dcal \to \mathbb{R}$ be a scalar-valued objective function and let $f$ be a continuous PWA function as in Definition~\ref{def:cpwa} that approximates $F$ with bounded approximation error $\delta = f - F$. Let $f_p$ in (\ref{eq:fp}) be the local convex segment of $f$ in its MMPS form (\ref{eq:fmmps}) on the set $\Ccal_{p,.}$, and let $\delta_p : \Ccal_{p,.} \to \mathbb{R}$ be the corresponding approximation error bounded by
	\[\sup_{x \in \Ccal_{p,.}} \vert \delta_p(x) \vert = \Delta_p < \infty.\]
	Let $x_p^\ast$ be any global minimizer of $f_p$ and $\hat{x}_p^\ast$ be any global minimizer of $F$ on $\Ccal_{p,.}$. Then, the following condition holds:
	\begin{align}\label{eq:maxdes}
		\Vert \hat{x}_p^\ast - x_p^\ast \Vert \leqslant \dfrac{2 \Delta_p }{c_1} + \max\limits_{q \in \Ical_{Q_p}} \left\{ \diam(\Ccal_{p,q})\right\},
	\end{align}
	where $c_1$ is defined in (\ref{eq:c1def}).
\end{thm}
\begin{proof}
	This can be directly concluded by extending Theorem~\ref{thm:phu} via considering Proposition~\ref{prop:h1is0} and Lemma~\ref{lem:h1form}.
\end{proof}

\section{Case Study}\label{sec:case}
To showcase the application of Theorem~\ref{thm:main}, we select the 1-dimensional cut of the well-known Eggholder test function~\cite{Whitley1996} at \mbox{$x_2 = 0$} given by
\begin{align*}
	F (x) &= -47 \sin \left(\sqrt{ \left\vert\frac{x}{2} + 47 \right\vert}\right) - x \sin \left( \sqrt{\vert x - 47 \vert }\right).
\end{align*}
We approximate $F$ by a continuous PWA function $f$ that we arbitrarily take as
\begin{align*}
	f (x) &= \min_{p \in \Ical_5} \left(f_{p,.}\right).
\end{align*}
with local convex segments
\begin{subequations}
	\begin{align}
		f_{1,.} & = \max_{q \in \Ical_3} \left(f_{1,q}\right), & \Ccal_{1,.} &= [-512,-385],\\
		f_{2,.} & = f_{2,1}, & \Ccal_{2,.} &= [-385,-330],\\
		f_{3,.} & = \max_{q \in \Ical_3} \left(f_{3,q}\right), & \Ccal_{3,.} &= [-330,-180],\\
		f_{4,.} & = \max_{q \in \Ical_3} \left(f_{4,q}\right) & \Ccal_{4,.} &= [-180,180],\\
		f_{5,.} & =  f_{5,1}, & \Ccal_{5,.} &= [180,512].
	\end{align}
\end{subequations}
Figure~\ref{fig:Fandf} shows the plots for the nonlinear objective function $F$ and its PWA approximation $f$. The subregions $\Ccal_{p,.}$ with $p \in \Ical_5$ are illustrated by different colors. 
Theorem~\ref{thm:main} can be used in two ways: 
\begin{enumerate}
	\item guaranteeing bounds on the distance between the regional minima of $F$ and $f$ on each subregion, given a bound on the approximation error, and
	\item finding the required criteria for the approximation to obtain a desired bound on the distance between these minima, which we refer to as the confidence radius.
\end{enumerate} 
We discuss each case separately by focusing on the approximation on $\Ccal_{3,.}$.

\subsection{Case 1: Finding the Confidence Radii}
Figure~\ref{fig:Ccal3F} shows $F$ and two PWA approximations on $\Ccal_{p,.}$ with two approximation errors. The first is $f_3^{(1)}$, which divides $\Ccal_{3,.}$ into 3 subregions with maximum approximation error $\Delta_3^{(1)} = 19.9$, and which is given by
\begin{align*}
	f_3^{(1)} (x) = \max \left\{ \begin{array}{c}
		-7.8 x -2365.7 \\ -0.9 x -501.2 \\ 6.1 x + 1176.1
	\end{array}\right\}.
\end{align*}
The second approximation is $f_3^{(2)}$ with 8 affine pieces, the maximum error $\Delta_3^{(2)} = 2.6$, and defined as
\begin{align*}
	f_3^{(2)} (x) = \max \left\{ \begin{array}{c}
		-8.6 x -2613.1 \\ -6.8 x -2095.6 \\ -4.6 x - 1477.9 \\ -2.2 x -829.8 \\ 0.3 x - 191.6 \\ 2.8 x + 412.5 \\ 5.1 x + 944 \\ 6.9 x + 1348.1
	\end{array}\right\}.
\end{align*}
The inverses of the corresponding convexity moduli are shown in Fig.~\ref{fig:Ccal3h1inv} in the same color as their corresponding $f$ in Fig.~\ref{fig:Ccal3F}, where $\chi$ is the confidence radius in Definition~\ref{def:chi}.

Using Theorem~\ref{thm:main}, the confidence radius for $f_3^{(1)}$ by is obtained by \mbox{$\chi^{(1)} = 70$}, which is the same value obtained by finding $h_1$ and its inverse function using~(\ref{eq:h1def}), which is presented in Fig.~\ref{fig:Ccal3h1inv}. The same process can be performed for the second approximation, $f_3^{(2)}$, which gives $\chi^{(2)} = 44.3$. Note that Theorem~\ref{thm:main} is more conservative for larger values of $\Delta$, compared to directly using the definition of $h_1$. For instance, if \mbox{$\Delta_3^{(2)} = 12.5$}, employing Theorem~\ref{thm:main} leads to $\chi^{(2)} = 65.7$, while the computed confidence radius using $h_1$ is $58.4$. The areas within the confidence radii for the PWA approximations are highlighted on the $x$-axis in Fig.~\ref{fig:Ccal3F} as well. 

\input{pics/pic_Ccal3}

\subsection{Case 2: Finding the Approximation Criteria}
In this case, we approach the problem from another direction: we select a desired confidence radius $\chi^{(3)}$ and find the required criteria for the corresponding approximated function, $f_3^{(3)}$. Let the desired \mbox{$\chi^{(3)} = 15$}; then, 
\[\dfrac{2 \Delta_3^{(3)} }{c_1} + \max\limits_{q \in \Ical_{Q_p}} \left\{ \diam(\Ccal_{3,q})\right\} \leqslant 10,\]
which means the diameter of the largest subregion $\Ccal_{3,.}$ must be smaller than 10. Firstly, given that \mbox{$\diam (\Ccal_{3,.}) = 150$}, it can be concluded that the PWA approximation requires at least 10 partitions. We can then start the approximation by partitioning $\Ccal_{p,.}$ into 15 subregions with the same diameter and find the lowest possible error bound $\Delta_3^{(3)}$ for the approximation, which is obtained as $2.83$ with $c_1 = 0.0072$. For this approximation, $\chi^{(3)}$ already exceeds $\diam(\Ccal_{3,.})$. 

To improve upon this example, we add another partition to reduce the largest partition diameter further and this time we do not aim at partitions of $\Ccal_{3,.}$ with the same diameter, but require 
\[\max\limits_{q \in \Ical_{16}} \{\diam \left( \Ccal_{3,q} \right)\} \leqslant 10. \]
We find $\Delta_3^{(3)} = 2.47$ with $c_1 = 1.03$ and 
\[\max\limits_{q \in \Ical_{16}} \{\diam \left( \Ccal_{3,q} \right)\} = 9.4. \]
For this values, we obtain $\chi^{(3)} = 14.24$. In case this value is acceptable, we can use the corresponding PWA approximation while ensuring that the minimizer of $F$ on $\Ccal_{3,.}$ lies in a ball or radius $14.24$ around the minimizer of $f^{(3)}$. In case a tighter confidence radius is desired, the same procedure can be followed by adding more subregions.

\subsection{Application for NMPC Optimization}
To showcase the application of our proposed method in PWA approximation to control optimization problems, we use an inverted-pendulum NMPC problem as in~\cite{Gros2020}. The considered prediction horizon is $N_{\mathrm{p}} = 2$ and the initial rotation velocity is set as $\dot{\theta} = 0$. The objective function $J_{\mathrm{NMPC}}$ can be formulated as a function of the measured pendulum angle $\theta_k$ and the control inputs $u_k$ and $u_{k+1}$. For instance, for $\theta_k = 0$ we have 
\begin{multline*}
	J_{\mathrm{NMPC}} (0, u_k, u_{k+1}) = 
	\sqrt{ \left( 0.02 u_k + \pi\right)^2+2 \pi^2} \\
	+ 0.02 \sqrt{u_k^2 + \left(u_k + u_{k+1}\right)^2}+0.01\sqrt{u_k^2+u_{k+1}^2}.
\end{multline*}
Moreover, the feasible region is defined as the box constraint $\vert u_{k+i-1}\vert \leqslant 20N$, $i \in \Ical_{2}$ with $\diam(\Ccal_{1,.}) = 56.4$. 

We approximate $J_{\mathrm{NMPC}}$ by two convex MMPS forms $f^{(1)}$ and $f^{(2)}$ -- with $P^{(1)} = P^{(2)} = 1$ in (\ref{eq:fmmps}) -- with different complexities in terms of the number of affine functions as

\begin{align*}
	&Q^{(1)} = 4, \quad \Delta^{(1)} = 0.19, \quad \max\limits_{q \in \Ical_{4}} \{\diam \left( \Ccal^{(2)}_{1,q} \right)\} = 28.2,\\
	&Q^{(2)} = 24, \quad \Delta^{(2)} = 0.01, \quad \max\limits_{q \in \Ical_{24}} \{\diam \left( \Ccal^{(2)}_{1,q} \right)\} = 14.6.
\end{align*}
\input{pics/pic_nmpc}
The inverse of the convexity modulus and the corresponding confidence regions for both approximations are shown in Fig.~\ref{fig:nmpc}. While $f^{(1)}$ has a low approximation error, its complexity level does not allow to guarantee a confidence radius lower than the diameter of the feasible region. However, the more accurate approximation $f^{(2)}$ guarantees a smaller confidence radius. Moreover, a general approximation criterion can be obtained, similar to the Eggholder NLP example, for an NMPC problem. In this case, it can be observed from (\ref{eq:maxdes}) and Fig.~\ref{fig:nmpch1} that $\chi$ is lower-bounded by the maximum subregion diameter. Therefore, if a particular confidence radius is desired, the approximation problem (\ref{eq:approxprob}) can be solved while imposing constraints on the diameter of subregions, e.g.\ an upper bound on the maximum subregion diameter.

\subsection{Discussion}
The Eggholder function and the NMPC case studies are two potential applications with, respectively, 1- and 3-dimensional domains to illustrate the theoretical results. Nevertheless, our proposed approach is applicable in higher dimensions since we do not assume any bounds on the dimension of $\Dcal$ throughout this paper. Moreover, the following aspects should be considered when applying the proposed approach:
\begin{itemize}
	\item The conservatism of the current approach in obtaining the confidence radii can be high, which can be seen in the case studies. 
	\item The complexity in determining the confidence radii is irrelevant to the dimension of the optimization problem. However, larger dimensions increase the computation time and range of $\gamma$ values for which (\ref{eq:h1def}) needs to be solved. Nevertheless, since this problem is solved offline, higher computation time do not limit the applicability of our approach.
\end{itemize}

\section{Conclusions}\label{sec:conc}

This paper has introduced a novel approach for bounding the minimizers of polytopically-constrained NLPs with nonlinear continuous objective functions using continuous PWA function approximations. We have leveraged the continuity of the PWA approximations resulting from employing an MMPS formalism to construct a locally-convex representation of the PWA approximation, thus facilitating the derivation of guaranteed bounds on the distance between the original and the approximated optimal solutions of the NLP by considering the maximal approximation error. Our approach offers a practical tool for determining criteria to achieve desired solution bounds. The effectiveness of the method has been demonstrated through two case studies on the Eggholder function and NMPC of an inverted pendulum, highlighting the practical application of the proposed method and its potential impact in optimization and optimal control.

For future work, our primary objectives are conducting an in-depth analysis of the conservatism of our approach, refining our sensitivity analysis, and extending our method to NLPs with non-convex constraints. Moreover, we aim to do more comparative studies to gain insight into the impacts of the improved computational efficiency through PWA approximation in light of the corresponding solution bounds. Finally, investigating the effects of probabilistic error bounds would be an interesting direction to help integrate our approach into learning-based and data-driven applications.

\bibliographystyle{unsrt}        
\bibliography{Citations}   

\end{document}

%% file: pics/pic_Fandf.tex
\begin{figure*}[htb]\centering
\begin{tikzpicture}
\begin{axis}[width=\textwidth,height=0.22\textwidth,xmin=-530,xmax=530,ymin=-600,ymax=600,
    xtick={-385,-330,-180,180},ytick={-400,300},axis lines=center,clip=false,
    xticklabel style={rotate=90},extra x ticks={-512,512},
    extra x tick style={tick label style={yshift=5mm,xshift=2.5mm,anchor=south}},
    legend columns = 4,legend style = {at={(0.5, -0.1)},anchor=north, inner sep=2pt,style={column sep=0.2cm},draw=none},
    legend cell align=center,]
    \begin{pgfonlayer}{background}
        \fill[color=RedOrange!10] (axis cs:-512,-600) rectangle (axis cs:-385,600);
        \fill[color=YellowOrange!10] (axis cs:-385,-600) rectangle (axis cs:-330,600);
        \fill[color=Green!10] (axis cs:-330,-600) rectangle (axis cs:-180,600);
        \fill[color=Cerulean!10] (axis cs:-180,-600) rectangle (axis cs:180,600);
        \fill[color=Plum!10] (axis cs:180,-600) rectangle (axis cs:512,600);
    \end{pgfonlayer}
    \node [right] at (current axis.right of origin) {$x$};
    \node [above] at (current axis.above origin) {$y$};
    \node[rectangle,draw=RedOrange!20,fill=RedOrange!10] at (-560,500) {\small$\Ccal_{1,.}$};
    \node[rectangle,draw=YellowOrange!20,fill=YellowOrange!10] at (-560,250) {\small$\Ccal_{2,.}$};
    \node[rectangle,draw=Green!20,fill=Green!10] at (-560,0) {\small$\Ccal_{3,.}$};
    \node[rectangle,draw=Cerulean!20,fill=Cerulean!10] at (-560,-250) {\small$\Ccal_{4,.}$};
    \node[rectangle,draw=Plum!20,fill=Plum!10] at (-560,-500) {\small$\Ccal_{5,.}$};
    \addlegendimage{black,line legend,thick}
    \addlegendentry{Nonlinear function $F$}
    \addlegendimage{NavyBlue,line legend,thick}
    \addlegendentry{PWA function $f$}			
    \addlegendimage{Red,only marks,mark=square*}
    \addlegendentry{Regional minima of $F$}			
    \addlegendimage{NavyBlue, only marks, mark=*}
    \addlegendentry{Regional minima of $f$}			
    \addplot [smooth,domain=-512:512,samples=100,color=black,thick,name path=NL]
    {-47*sin(deg(sqrt(abs(0.5*x+47))))-x*sin(deg(sqrt(abs(x-47))))};
    \draw[thick,NavyBlue,smooth] (-512,-550)--(-490,-500)--(-400,270)--(-385,385)--(-330,208)--(-270,-250)--(-240,-280)--(-180,78)--(-66,-40)--(90,-20)--(180,190)--(512,-250);
    \draw[thin,densely dashed,Red] (-512,-550) 
    coordinate[rectangle,inner sep=2.5pt,fill=Red] node[above,Red] {$\hspace{4mm}\hat{x}^\ast_1$} -- (-512,0);
    \draw[thin,densely dashed,Red] (-330,220) 
    coordinate[rectangle,inner sep=2.5pt,fill=Red] node[left,Red] {$\hspace{4mm} \hat{x}^\ast_2$} -- (-330,0);
    \draw[thin,densely dashed,Red] (-250,-275) 
    coordinate[rectangle,inner sep=2.5pt,fill=Red] node[below,Red] {$ \hat{x}^\ast_3 $} -- (-250,0);
    \draw[thin,densely dashed,Red] (-90,-110) 
    coordinate[rectangle,inner sep=2.5pt,fill=Red] node[below,Red] {$\hat{x}^\ast_4 \hspace{3mm}$} -- (-90,0);
    \draw[thin,densely dashed,Red] (467,-420) 
    coordinate[rectangle,inner sep=2.5pt,fill=Red] node[below,Red] {$\hat{x}^\ast_5$} -- (467,0);
    \draw[thin,densely dashed,NavyBlue] (-512,-550) 
    coordinate[circle,inner sep=1.5pt,fill=NavyBlue] node[below,NavyBlue] {$x^\ast_1$} -- (-512,0);
    \draw[thin,densely dashed,NavyBlue] (-330,280) 
    coordinate[circle,inner sep=1.5pt,fill=NavyBlue] node[right,NavyBlue] {$x^\ast_2$} -- (-330,0);
    \draw[thin,densely dashed,NavyBlue] (-240,-280) 
    coordinate[circle,inner sep=1.5pt,fill=NavyBlue] node[right,NavyBlue] {$x^\ast_3$} -- (-240,0);
    \draw[thin,densely dashed,NavyBlue] (-66,-40) 
    coordinate[circle,inner sep=1.5pt,fill=NavyBlue] node[below,NavyBlue] {$x^\ast_4$} -- (-66,0);
    \draw[thin,densely dashed,NavyBlue] (512,-250) 
    coordinate[circle,inner sep=1.5pt,fill=NavyBlue] node[below,NavyBlue] {$x^\ast_5$} -- (512,0);
\end{axis}
\end{tikzpicture}
\caption{Plots of the nonlinear objective function $F$ and its PWA approximation $f$.}
\label{fig:Fandf}
\end{figure*}
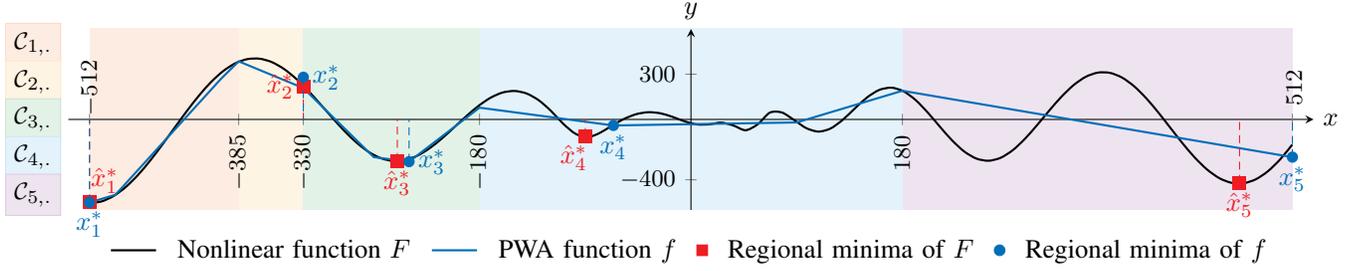

%% file: pics/pic_Ccal3.tex
\vspace{-10pt}
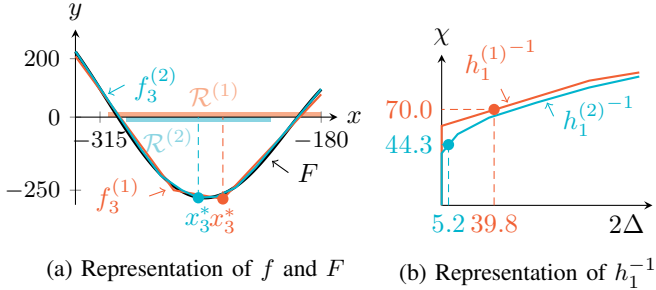
\begin{figure}[htbp]
\begin{center}
\begin{subfigure}[t]{0.52\linewidth}\centering
\begin{tikzpicture}	\hspace{-10pt}	
	\begin{axis}[width=1.1\textwidth,height=0.85\textwidth,xmin=-330,xmax=-170,ymin=-300,ymax=300,
		xtick={-315,-180},ytick={-250,0,200},
		axis x line=center,clip=false,axis y line=left,
		legend cell align=center,]
		\begin{pgfonlayer}{background}
			\fill[color=Turquoise!50] (axis cs:-299,-15) rectangle node[below,xshift=-2.5ex] {$\Rcal^{(2)}$} (axis cs:-211,0);
			\fill[color=RedOrange!50] (axis cs:-310,0) rectangle node[above] {$\Rcal^{(1)}$} (axis cs:-180,15);
		\end{pgfonlayer}
		\node [right] at (current axis.right of origin) {$x$};
		\node [above] at (axis cs:-330,300) {$y$};
		\addplot [smooth,domain=-330:-180,samples=50,color=black,very thick,name path=NL]
		{-47*sin(deg(sqrt(abs(0.5*x+47))))-x*sin(deg(sqrt(abs(x-47))))};
		\draw[thick,RedOrange,smooth] (-330,208)--(-270,-250)--(-240,-280)--(-180,78);
		\draw[thick,Turquoise,smooth] (-330,225)--(-290,-110)--(-278,-200)--(-265,-250)--(-255,-276)--(-244,-272)--(-232,-240)--(-220,-172)--(-180,96);
		\draw[thin,densely dashed,RedOrange] (-240,-280) 
		coordinate[circle,inner sep=1.5pt,fill=RedOrange] node[below,RedOrange] {$x^\ast_3$} -- (-240,0);
		\draw[thin,densely dashed,Turquoise] (-255,-276) 
		coordinate[circle,inner sep=1.5pt,fill=Turquoise] node[below,Turquoise] {$x^\ast_3$} -- (-255,0);
        \draw[<-,RedOrange] (-275,-250) -- (-285,-270) node[left] {$f_3^{(1)}$};
        \draw[<-,Turquoise] (-310,80) -- (-300,110) node[right] {$f_3^{(2)}$};
        \draw[<-,black] (-210,-150) -- (-200,-190) node[right] {$F$};
	\end{axis}
\end{tikzpicture}
\subcaption{Representation of $f$ and $F$}\label{fig:Ccal3F}
\end{subfigure}			
\begin{subfigure}[t]{0.45\linewidth}\centering
\begin{tikzpicture}		\hspace{-5pt}
	\begin{axis}[width=1.1\textwidth,height=0.9\textwidth,xmin=0,xmax=160,ymin=0,ymax=110,
		xtick=\empty,ytick=\empty,
		axis x line=bottom,clip=false,axis y line=left,]
		\node [right,yshift=-7pt,xshift=-20pt] at (current axis.right of origin) {$2 \Delta$};
		\node [above] at (axis cs:0,110) {$\chi$};
		\draw[thick,RedOrange,smooth] (0,0)--(0,58)--(112,91)--(150,97);
		\draw[thick,Turquoise,smooth] (0,0)--(0,38)--(12.2,52)--(36,64)--(72,75)--(112,86)--(150,94);
		\draw[thin,densely dashed,RedOrange] (39.8,0) node[below] {$39.8$} -- (39.8,70) 
		coordinate[circle,inner sep=1.5pt,fill=RedOrange] -- (0,70) node[left] {$70.0$};
		\draw[thin,densely dashed,Turquoise] (5.2,0) node[below] {$5.2$} -- (5.2,44.3) 
		coordinate[circle,inner sep=1.5pt,fill=Turquoise] -- (0,44.3) node[left] {$44.3$};	
        \draw[<-,RedOrange] (50,75) -- (45,85) node[above] {${h_1^{(1)}}^{-1}$};
        \draw[<-,Turquoise] (80,75) -- (85,65) node[right] {${h_1^{(2)}}^{-1}$};
	\end{axis}
\end{tikzpicture}
\subcaption{Representation of $h_1^{-1}$}\label{fig:Ccal3h1inv}
\end{subfigure}			
\caption{Comparison of two different PWA approximations of the nonlinear function on $\Ccal_{3,.}$.}\label{fig:Ccal3}
\end{center}
\end{figure}
\vspace{-5pt}

%% file: pics/pic_nmpc.tex
\begin{figure}[htbp]
\begin{center}
\begin{subfigure}[t]{0.28\textwidth}\centering
\begin{tikzpicture}		
    \begin{axis}[width=\textwidth,height=0.7\textwidth,xmin=0,xmax=1,ymin=0,ymax=65,
    xtick={0.02,0.89},ytick={14.6,28.3,56.4},xticklabel style={/pgf/number format/fixed},
    axis x line=bottom,clip=false,axis y line=left,legend style = {at={(0.6, 0.7)},anchor=north, inner sep=2pt,style={column sep=0.2cm},draw=none,fill=none},
    legend cell align=center,]
    \draw[thick,NavyBlue,smooth] (0,0)--(0,28.3)--(0.19,44)--(0.29,47)--(0.55,51)--(0.73,54)--(0.89,56.4);
    \node [right] at (current axis.right of origin) {$2 \Delta$};
    \node [above] at (axis cs:0,66) {$\chi$};
    \draw[thin,densely dashed,NavyBlue] (0,56.4) -- (0.89,56.4) -- (0.89,0); 
    \draw[thick,Maroon,smooth] (0,0)--(0,12.1)--(0.10,25)--(0.19,33)--(0.27,42)-- (0.37,48) --(0.76,56.4);
    \draw[thin,densely dashed,Maroon] (0,14.6) -- (0.02,14.6) -- (0.02,0); 
    \addlegendimage{NavyBlue,line legend,thick}
    \addlegendentry{${h_1^{-1}}^{(1)}$};	
    \addlegendimage{Maroon,line legend,thick}
    \addlegendentry{${h_1^{-1}}^{(2)}$};
    \end{axis}
\end{tikzpicture}
\subcaption{Confidence radii}\label{fig:nmpch1}
\end{subfigure}	
\begin{subfigure}[t]{0.2\textwidth}\centering
\begin{tikzpicture}		
    \begin{axis}[width=\textwidth,height=\textwidth,xmin=-35,xmax=35,ymin=-35,ymax=35,xtick={0},ytick={0},xticklabel style={/pgf/number format/fixed},axis x line=middle,clip=false,axis y line=middle,]
    \draw[NavyBlue,thick,fill=NavyBlue!10,fill opacity=0.7] (axis cs:0,0) circle [radius=28.2];
    \draw[Maroon,thick,fill=Maroon!10,fill opacity=0.7] (axis cs:0,0) circle [radius=14.6];
    \draw[thick,black,smooth] (-20,-20)--(-20,20)--(20,20)--(20,-20)--(-20,-20);
    \node [right] at (current axis.right of origin) {$u_k$};
    \node [above] at (axis cs:0,35) {$u_{k+1}$};
    \addplot[only marks,mark=x] coordinates {(0,0)};
    \end{axis}
\end{tikzpicture}
\subcaption{Confidence regions}\label{fig:nmpch1reg}
\end{subfigure}	
\caption{Comparison of two PWA approximations of NMPC.}\label{fig:nmpc}
\end{center}
\end{figure}
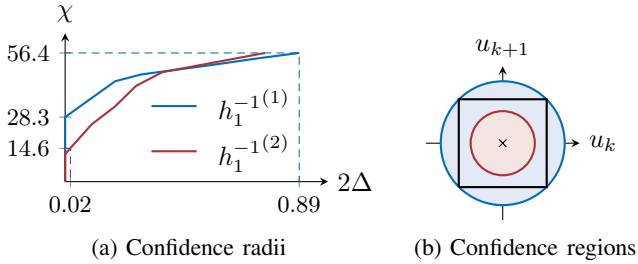